\documentclass[11pt,a4paper]{article} 
\usepackage{color} 
\usepackage[latin1]{inputenc} 
\usepackage[T1]{fontenc} 
\usepackage{amsfonts}
\usepackage{amssymb}
\usepackage{url}
\usepackage{graphicx}
\usepackage{dsfont}
\usepackage{epsfig}
\usepackage{epstopdf}
\usepackage{ae}
\usepackage{amsmath}
\usepackage{endnotes}
\usepackage{setspace}
\usepackage{amsthm}
\usepackage{appendix}
\usepackage{natbib}%
\providecommand{\U}[1]{\protect\rule{.1in}{.1in}}
\providecommand{\U}[1]{\protect\rule{.1in}{.1in}}
\theoremstyle{plain}
\newtheorem{theo}{Theorem}[section]
\newtheorem{proposition}{Proposition}[section]
\newtheorem{definition}{Definition}[section]
\theoremstyle{definition}

\theoremstyle{remark}
\newtheorem{rem}{Remark}[section]

\begin{document}

\title{{\LARGE \textbf{Algorithms for Finding Copulas \\\vspace{2mm}
Minimizing Convex Functions of
Sums}}}
\author{\textsc{Carole Bernard} \thanks{ C. Bernard is with the department of
Accounting, Law and Finance at Grenoble Ecole de Management, Email \texttt{
carole.bernard@grenoble-em.com}. } and \textsc{Don McLeish} \thanks{D.L. McLeish
is with the department of Statistics and Actuarial Science at the University
of Waterloo, Email \texttt{dlmcleis@uwaterloo.ca}. }}
\date{\today}
\maketitle

\begin{abstract} 
We develop improved rearrangement algorithms to find the dependence structure
that minimizes a convex function of the sum of dependent variables with given
margins. We propose a new multivariate dependence measure, which can assess the
convergence of the rearrangement algorithms and can be used as a stopping
rule. We show how to apply these algorithms for example to finding the dependence among variables for which the marginal distributions and the distribution of the sum or the difference are known. As an example, we can find the dependence between two uniformly distributed variables that makes the
distribution of the sum of two uniform variables indistinguishable from a
normal distribution. Using MCMC techniques, we  design an algorithm that converges to the global optimum. 
\end{abstract}

Key-words: Multivariate risk measure, Block RA, Rearrangement algorithm, minimum variance, discrete optimization, MCMC, copula.

\newpage

\begin{center}
{\LARGE \textbf{Algorithms for Finding Copulas\\\vspace{4mm} Minimizing
Convex Functions of Sums}}
\end{center}

\vspace{2cm}

\section*{Introduction}

For specified marginal distributions such as the uniform or the normal distribution, can we
find a dependence structure or copula, which provides a specific distribution
for the sum of $n$ variables? What if we were to require that the sum
be constant? Questions like this have been addressed theoretically in the
literature in a number of papers 
with the concept of complete mixability (\cite{WW11}, \cite{puccetti2014general}, \cite{wang2015joint} to cite only a few), and computationally,
with the rearrangement algorithm (RA) (\cite{Puccetti-Rueschendorf-2012},
\cite{Embrechts-et-al-2013}). The RA aims to minimize the expectation of a convex
function of a sum of random variables (including the case of minimization of
the variance of the sum as a special case).\footnote{More details on the RA can be found at \url{https://sites.google.com/site/rearrangementalgorithm/}.} This algorithm is fast and simple
but may not converge to the global minimum. In particular, it does not depend
on the convex function to minimize. Our main objective in this paper is to
further this discussion by developing an improved version of this algorithm.

The minimization of convex functions of a sum of dependent random variables can be
formulated using a matrix $\mathbf{X}:=(X_{ij})_{i,j}$, and is linked to the problem of
minimizing a convex function of the row sums $\sum_{j=1}^{n}X_{ij}$ over all
permutations within the columns. It is a highly computationally complex
problem, as even in the special case of $n=3$ columns, it has been shown to be
NP-complete (\cite{Haus}). It means that no algorithm will guarantee
convergence to the optimum in polynomial time. Enumeration might be
considered, but for reasonable sized matrices this is also not feasible. For a
$m\times n$ matrix, the number of essentially distinct matrices that can be
obtained is the number of permutations of the columns other than the first,
which is $(m!)^{n-1}.$ For example, for a small $10\times6$ matrix this is
$(10!)^{5}=6.292\,4\times10^{32},$ obviously completely impossible by
enumeration. Various versions of this problem have been treated in the
literature, and algorithms proposed for special cases, but because the problem
is NP complete (see \cite{Hsu}), these algorithms do not converge to an
optimal. For example \cite{CY} propose an
algorithm designed to minimize the maximum of the row sums (the
\textit{assembly line crew scheduling} problem) and show that this algorithm
converges to a solution less than 1.5 times the optimal. The NP complete nature of
the problem indicates that the worst case performance of algorithms may be
unsatisfactory, but it is still possible to develop algorithms, which normally
find the optimum very quickly. We  present Markov Chain algorithms here, which
guarantee finding the global optimum in finite time, usually very rapidly. 

There are many more applications for the Rearrangement Algorithm (RA) of
\cite{Puccetti-Rueschendorf-2012}. It has been used successfully in recent
advances to the risk management field. Specifically, the RA is used to measure
model risk on dependence, also called \textquotedblleft dependence
uncertainty,'' and can help regulators to make decisions on
which risk measure is most appropriate to compute capital requirements (see
\cite{Embrechts-et-al-2014}). It was successfully used to approximate VaR
bounds on the sum of $n$ dependent risks with given marginal distributions by
\cite{Embrechts-et-al-2013} by applying the RA to the largest rows of the
matrix such that all risks are comonotonic. Many more applications of this RA have
been recently developed. Among others, \cite{aas2013bounds} use the standard
RA to compute capital requirements of DNB bank,
\cite{Bernard-Ru-Vanduffel-2014}, \cite{bernard2015robust} to assess portfolios' credit risk and \cite{Bernard-Vanduffel-2014}
to incorporate partial information on dependence in the computation of bounds
on capital requirements, \cite{BJW14} to derive bounds on convex risk measures
and quantify dependence uncertainty and \cite{puccettidetecting} to detect
complete mixability. As we expect many more applications of such algorithms,
it is important to develop efficient and accurate algorithms that converge to
a minimum, a feature of our algorithms.

In this paper, we develop a novel application. We  are able to find the dependence structure such that the sum of
dependent variables has a prescribed distribution. We will refer to two
distributions $F$ and $G$ as \textquotedblleft close" if a large sample from
one, say from $G$, cannot be detected as \textit{not} coming from $F$ with high
probability using a standard test. We will use the Kolmogorov-Smirnov test and
the Wasserstein distance test. This allows us, for example, to find two
dependent Uniform[0,1] random variables $U_{i}$ such that $U_{1}%
+U_{2}-Z_{m,\sigma^{2}}$ is nearly 0 where $Z_{m,\sigma^{2}}$ is a normally
distributed variable with mean $m$ and variance $\sigma^{2}$, in other words
so that $U_{1}+U_{2}$ is nearly normal.

This toy example illustrates a potential use of the methodology to infer the dependence among variables that can explain a given distribution for the sum. At first, it may sound limited and more a mathematical curiosity but this methodology can also be useful in practice. For example, it can be  used in finance.  Assuming that prices of basket options or spread options are available at the same time as prices of regular options written on individual stocks, our methodology can then be useful to infer a multivariate model for the assets that is consistent with this information. It is then  particularly interesting in fitting the bivariate distribution between  gas and electricity prices given that spark spread options are actively traded (see  \cite{alexander2004bivariate}, \cite{carmona2003pricing}, \cite{rosenberg2000nonparametric}).


The paper is organized as follows. In Section \ref{S1}, we present a new
multivariate measure inspired by the notion of $\Sigma$-countermonotonicity
discussed by \cite{puccetti2014general}. Then, in Section \ref{S2}, we develop
improved rearrangement algorithms that may use this multivariate dependence
measure as a stopping rule and discuss their relative performance. Our
algorithms can converge in fewer steps by selecting the blocks optimally and
converge to a point much closer to the global optimum than the standard RA,
often by orders of magnitude.  Section \ref{S3} illustrates the methodology
with the explicit construction of the dependence that makes the sum of two
uniformly distributed variables indistinguishable from a normal distribution.
More generally, we are interested in whether a copula exists such that the sum
of $m$ random variables from one distribution has another prescribed distribution. We then briefly discuss an application to finance. Finally, in Section \ref{SMCMC}, we show how to modify
the block RA using Markov Chain Monte Carlo methods to guarantee finding the
global minimum in finite time and accounting for a given convex measure of the sum.

\section{A new multivariate measure of dependence\label{S1}}

In this section, we propose to extend any dependence measure defined between
two random variables to a multivariate dependence measure in a natural way.

\subsection{A new multivariate measure based on $\Sigma$-countermonotonicity}

This multivariate measure will play a crucial role in assessing the
convergence of the rearrangement algorithm that minimizes the variance of the
sum of dependent risks with given marginals (\cite{Puccetti-Rueschendorf-2012}
and \cite{Embrechts-et-al-2013}). It is inspired by the recent notion of
$\Sigma$-countermonotonicity introduced by \cite{puccetti2014general} in
which all sums over disjoint subsets $\Pi$ and $\bar\Pi$ such that $\Pi
\cup\bar\Pi=\{1,2,...,n\}$ are countermonotonic (see also \cite{lee2014multidimensional}).

\begin{definition}
Let $\phi\left(  \mathbf{X}_{1},\mathbf{X}_{2}\right)  $ be a measure of
dependence between two columns of data $\mathbf{X}_{1}$ and $\mathbf{X}_{2}$
such as Spearman's rho, Kendall's tau, or Pearson correlation coefficient. For
a matrix of data $\mathbf{X}=[\mathbf{X}_{1},\mathbf{X}_{2},...,\mathbf{X}%
_{n-1},\mathbf{X}_{n}]$ with $n$ columns, we define the multivariate measure
of dependence
\begin{equation}
\varrho(\mathbf{X}):=\frac{1}{2^{n-1}-1}\sum_{\Pi\in\mathcal{P}}\phi\left(
\sum_{i\in\Pi}\mathbf{X}_{i},\sum_{i\in\bar{\Pi}}\mathbf{X}_{i}\right)
\label{measure1}%
\end{equation}
where the sum is over the set $\mathcal{P}$ consisting of $2^{n-1}-1$ distinct
partitions of $\{1,2,...,n\}$ into \textbf{non-empty} subsets $\Pi$ and its
complement $\bar{\Pi}.$\footnote{There are $2^{n}$ partitions $(\Pi,\bar{\Pi
})$ so that $\{1,2,...,n\}=\Pi\cup\bar{\Pi}$ and $\Pi\cap\bar{\Pi}=\emptyset$.
But the measure $\phi\left(  \sum_{i\in\Pi}\mathbf{X}_{i},\sum_{i\in\bar{\Pi}%
}\mathbf{X}_{i}\right)  $ is usually meaningless when either $\Pi$ or
$\bar{\Pi}$ are empty. Moreover the partition ($\Pi$ , $\bar{\Pi})$ is
essentially counted twice. So there are $(2^{n}-2)/2$ relevant
\textquotedblleft distinct\textquotedblright\ partitions into non-empty subsets.}
\end{definition}

For the remainder of the paper, we assume that $\phi$ is the Spearman
correlation. Let us recall its definition for two continuous random variables
$X$ and $Y$ with respective marginal c.d.f. $F_{X}$ and $F_{Y}$. The Spearman
correlation is then equal to
\begin{equation}
\label{SRho}\phi(X,Y):=\frac{\hbox{cov}(F_{X}(X),F_{Y}(Y))}{\sqrt
{\hbox{var}(F_{X}(X))\hbox{var}(F_{Y}(Y))}},
\end{equation}
which corresponds to the correlation between the two uniformly distributed
generators of $X$ and $Y$ respectively. The results using alternatives such as
Kendall's tau would be similar. The minimum Spearman correlation is -1 and it
is achieved by the countermonotonicity structure (originally called
``antithetic" dependence in the language of \cite{hammersley1964monte}).

This measure $\varrho(\mathbf{X})$ is different from the multivariate
Kendall's tau, multivariate Spearman correlation, the average pairwise
Kendall's tau, or the average pairwise Spearman correlation recalled in
Definition 8 of \cite{lee2014multidimensional}. In \eqref{measure1}, we
average the bivariate dependence measure $\phi$ between two sums taken over
the two subsets $\Pi$ and $\bar{\Pi}$ of a partition of $\{1,2,...,n\},$ i.e.
two disjoint non-empty sets $\Pi$ and $\bar{\Pi}$ with $\Pi\cup\bar{\Pi
}=\{1,2,...,n\}$. Contrary to existing multivariate dependence measures, it is not driven by the dependence pairwise. In addition, it has a nice connection with convex order as outlined in Remark \ref{res} hereafter.

This $n$-dimensional dependence measure, $\varrho(\mathbf{X}),$ can be
unbiasedly estimated either by choosing some of the $2^{n-1}-1$ such
partitions without replacement or by assigning columns at random, e.g. put
\[
\mathbf{S}_{n}=\mathbf{\mathbf{X}_{1}+\mathbf{X}_{2}+...+\mathbf{X}%
_{n-1}+\mathbf{X}_{n}}%
\]
and average the values of
\begin{equation}
\phi\left(  \sum_{i=1}^{n}I_{i}\mathbf{X}_{i},\mathbf{S}_{n}-\sum_{i=1}^{n}
I_{i}\mathbf{X}_{i}\right)  \label{estimate1}%
\end{equation}
over many samples of random independent Bernoulli variables $I_{i}$ for
which $0<\sum_{i=1}^{n}I_{i}<n$. \

\begin{rem}
As a side remark, we give the continuous formulation of our newly proposed
multivariate risk measure $\varrho.$ Starting from the definition of the
Spearman correlation in \eqref{SRho}, and using the moments of a uniformly
distributed variable, $\hbox{var}(F_{X}(X))=\frac{1}{12}$ and $E(F_{X}%
(X))=\frac{1}{2},$
\begin{align*}
\phi(X,Y)  &  =12\int_{0}^{1}\int_{0}^{1}P(F_{X}(X)>x,F_{Y}(Y)>y)dxdy-3
\end{align*}
(see for example \cite{N}). Therefore, with $S_{\Pi}=\sum_{i\in\Pi}%
\mathbf{X}_{i}$ we define the multivariate measure of dependence related to
the joint distribution of a random vector $\mathbf{X}$ of dimension $n,$
\begin{align*}
\varrho(\mathbf{X})  &  :=\frac{1}{2^{n-1}-1}\sum_{\Pi\in\mathcal{P}}%
\phi\left(  S_{\Pi},S_{\bar{\Pi}}\right) \\
&  =12\int_{0}^{1}\int_{0}^{1}\left[  \frac{1}{2^{n-1}-1}\sum_{\Pi
\in\mathcal{P}}P\left(  F_{S_{\Pi}}(S_{\Pi})>x,F_{S_{\bar{\Pi}}}(S_{\bar{\Pi}%
})>y\right)  \right]  dxdy-3
\end{align*}
This can be estimated unbiasedly by choosing one or more partitions $\Pi$ and
$\bar{\Pi}$ at random in the set $\mathcal{P}$ of all possible $2^{n-1}-1$
partitions and corresponding uniformly distributed random numbers
$U,V\sim\mathcal{U}[0,1]$ and using 12 times the proportion of times that
$F_{S_{\Pi}}(S_{\Pi})>U$ and $F_{S_{\bar{\Pi}}}(S_{\bar{\Pi}})>V$ minus 3.
\end{rem}

\subsection{Necessary condition to minimize convex functions of a sum}

It has been noted in \cite{Puccetti-Rueschendorf-2012} that the situation in
which all the columns are countermonotonic with the sum of all others is a
necessary condition to have a dependence structure that minimizes the
expectation of a convex function of a sum. Proposition \ref{PP1} below is a
straightforward extension. The result holds for the minimization of any
expectation of a convex function and as a special case for the variance of the
sum $\hbox{var}(\sum X_{i})$. We provide a counterexample to show that the
condition is not sufficient.

\begin{proposition}
[Necessary condition to minimize expected convex functions of a sum]
~\newline\label{PP1} Let $f$ be a convex function. If $\text{E}\left(
f\left(  \sum_{i}\mathbf{X}_{i}\right)  \right)  $ is at a minimum then
$\phi\left(  \sum_{i\in\Pi}\mathbf{X}_{i},\sum_{i\in\bar{\Pi}}\mathbf{X}%
_{i}\right)  $ is minimized for every partition into two sets $\Pi$ and
$\bar{\Pi}.$ However, the converse does not hold in general.
\end{proposition}

\proof The sufficient condition is proved in Theorem 3.8 (d) of
\cite{puccetti2014general}.
The other direction is unfortunately false. For example, consider the matrix
below:
\begin{equation}
B_{1}=\left(
\begin{array}
[c]{cccc}%
0.0662 & 0.2571 & 0 & -0.5842\\
0.3271 & 1.0061 & -1.3218 & -0.0833\\
0.6524 & -0.6509 & -0.0549 & 0.2495\\
1.0826 & -0.9444 & 0.9248 & -0.9263
\end{array}
\right)  \label{B2}%
\end{equation}
It is straightforward to check (with basic calculations) that for all $7$ possible partitions $\Pi,\bar{\Pi}$ we have
that $S_{\Pi},S_{\bar{\Pi}}$ are countermonotonic so that $\phi\left(
\sum_{i\in\Pi}\mathbf{X}_{i},\sum_{i\in\bar{\Pi}}\mathbf{X}_{i}\right)  =-1$
for all such partitions. The variance of the row sums is 0.04346. However, the
matrix
\begin{equation}
B_{2}=\left(
\begin{array}
[c]{cccc}%
0.0662 & 1.0061 & -1.3218 & 0.2495\\
0.3271 & 0.2571 & 0 & -0.5842\\
0.6524 & -0.6509 & 0.9248 & -0.9263\\
1.0826 & -0.9444 & -0.0549 & -0.0833
\end{array}
\right)  \label{B1}%
\end{equation}
obtained by a slightly different permutation of the columns provides constant
($=0$) row sums with a strictly smaller value of the variance of the row sums
(as the variance is then equal to zero). It is a counterexample for the expectation of any convex function and
not just for the variance. \hfill$\Box$

\begin{rem}\label{res}
Recall that the Spearman correlation $\phi$ between $S_{\Pi
}$ and $S_{\bar{\Pi}}$ is minimized with the value -1
achieved by the countermonotonicity between the pair of sums $S_{\Pi
}$ and $S_{\bar{\Pi}}$. Therefore, Proposition \ref{PP1} shows that a necessary
condition to attain a dependence between ${\mathbf{X}}_{i}$ that minimizes the
expectation of a convex function and thus the variance is that
\[
\varrho(\mathbf{X})=-1.
\]
\end{rem}

Note also that Proposition \ref{PP1} can be applied more generally to supermodular functions and convex functions of a sum are only special cases.

\section{Improved Rearrangement Algorithms\label{S2}}

In this section, we start by recalling the standard rearrangement algorithm
(RA) of \cite{Puccetti-Rueschendorf-2012} and \cite{Embrechts-et-al-2013}. We
then show how to improve it by designing the Block RA. We then illustrate the
improvement through some numerical examples. To facilitate the exposition in
this section, we develop algorithms aimed at minimizing the variance of the
sum. In Section \ref{SMCMC}, we will show how to adapt these algorithms to
ensure convergence to the global minimum of the expectation of a specific
convex function of the sum that is not necessarily the variance.

\subsection{Standard Rearrangement Algorithm}

The standard rearrangement algorithm is a method of constructing dependence
between variables $X_{j}$ $(j=1,2,\dots,n)$ such that the variance of the sum
$S_{n}$ becomes as small as possible. Consider a matrix $\mathbf{X}=[x_{{ij}%
}]_{i,j}$, corresponding to a multivariate vector $[\mathbf{X}_{1}%
,\mathbf{X}_{2},\dots,\mathbf{X}_{n}]$.

\paragraph{Standard Rearrangement Algorithm}

For $i$ from 1 to $n$, make the $i^{th}$ column countermonotonic with the sum
of the other columns. Repeat this process (by starting again from the first
column) until each column is countermonotonic with the sum of the other columns.

At each step of this algorithm, we make the $j^{th}$ column countermonotonic
with the sum $\sum_{i\neq j}X_{i}$, so that the variance of the sum of all
columns before rearranging is larger than the variance of the sum of all
columns after rearranging. At each step of the algorithm the variance
decreases, it is bounded from below (by 0) and thus converges (given that
there is a finite number of permutations of rows and columns). If it gets to
0, we have found a perfect mixability situation in which the dependence makes
the sum constant. Otherwise, there is no guarantee that we have found the
global minimum of the variance of the sum over all dependence structures.

We note however that it is possible to converge to a matrix $\mathbf{X}$ for
which $\varrho(\mathbf{X})>-1$ and therefore, that does not satisfy the
necessary condition of Proposition \ref{PP1}. For example, consider the
following matrix
\begin{equation}
\label{matC}C=\left[
\begin{array}
[c]{cccc}%
1.1423 & 0.3674 & 1.8266 & 2.1637\\
1.9135 & 0.9880 & 0.5237 & 2.0392\\
2.8994 & 0.0377 & 1.5924 & 1.0061\\
4.0077 & 0.8852 & 0.1974 & 0.4097
\end{array}
\right]  .
\end{equation}
The matrix $C$ is such that $S_{\Pi,}$ and $S_{\bar{\Pi}}$ are
countermonotonic whenever $\Pi=\{i\}$. In this case, the multivariate
dependence measure is $\varrho(\mathbf{X})=-0.9714$ because certain subsets
are not countermonotonic, in particular $S_{\Pi}$ and $S_{\bar{\Pi}}$ with
$\Pi=\{1,3\}$ and $\{2,3\}.$ For this matrix, the standard RA has already
declared convergence.

\subsection{Block Rearrangement Algorithm}

We now construct a version of the rearrangement algorithm designed to reduce
the measure $\varrho(\mathbf{X})$ in order to improve the convergence to the
minimum variance. Suppose for each partition $\Pi\in\mathcal{P}$ we know the
values of $\rho_{\Pi}=\phi\left(  \sum_{i\in\Pi}\mathbf{X}_{i},\sum_{i\in
\bar{\Pi}}\mathbf{X}_{i}\right).$ In order to reduce the variance of the
sum, we wish to reduce the covariances $\hbox{cov}\left(  \sum_{i\in\Pi}%
\mathbf{X}_{i},\sum_{i\in\bar{\Pi}}\mathbf{X}_{i}\right)  $ and in particular,
rearrange so as to reduce the largest of these values. We will therefore apply
a rearrangement of the elements of $\mathbf{X}_{i},i\in\bar{\Pi}$ so that the
sums $\sum_{i\in\bar{\Pi}}\mathbf{X}_{i}$ are countermonotonic to $\sum
_{i\in\Pi}\mathbf{X}_{i}.$

Suppose that the matrix $\mathbf{X}=[\mathbf{X}_{1},\mathbf{X}_{2}%
,...,\mathbf{X}_{n-1},\mathbf{X}_{n}]$ has covariance matrix $\Sigma.$ Note
that $S_{\Pi}=\sum_{i\in\Pi}\mathbf{X}_{i}$ and so
\[
\text{var}\left(  \sum_{i}\mathbf{X}_{i}\right)  =\text{var}(S_{\Pi
})+\text{var}(S_{\bar{\Pi}})+2\text{cov}(S_{\Pi},S_{\bar{\Pi}})
\]
This consists of the sum of three classes of elements of the covariance matrix:

(a) the sum of $\Sigma_{ij}$ for both $i,j\in\Pi$

(b) the sum of $\Sigma_{ij}$ for both $i,j\in\bar{\Pi}$

(c) the sum of $\Sigma_{ij}$ for $i\in\Pi,j\in\bar{\Pi}.$

An algorithm which proceeds at each step by keeping the values of
$\text{var}(S_{\Pi})$, $\text{var}(S_{\bar{\Pi}})$ constant while minimizing
the value of $\text{cov}(S_{\Pi},S_{\bar{\Pi}})$ over rearrangements of the
blocks, is bound to result in a non-increasing variance and will therefore
converge. In order to obtain a maximum benefit from this single rearrangement,
we wish to choose a subset ${\Pi}$ for which the Spearman correlation
$\phi\left(  S_{\Pi},S_{\bar{\Pi}}\right)  $ is the largest and then rearrange
the second block so that $S_{\bar{\Pi}}$ is countermonotonic to the values of
$S_{\Pi}$, thereby rendering $\phi\left(  S_{\Pi},S_{\bar{\Pi}}\right)  =-1.$
Since $\text{var}(S_{\Pi}),\text{var}(S_{\bar{\Pi}})$ are unchanged and
$\text{cov}(S_{\Pi},S_{\bar{\Pi}})$ is reduced, this results in a reduction of
$\text{var}(\sum_{i}\mathbf{X}_{i}).$ It turns out that choosing the largest
Spearman correlation $\phi\left(  S_{\Pi},S_{\bar{\Pi}}\right)  $ among a
relatively small number of possible partitions speeds up the algorithm and is
adequate. For a matrix $\mathbf{X}$ with $n$ columns, there are $p:=2^{n-1}-1$
possible subsets of $\Pi\subset\{1,2,3,...,n\}$ such that $\bar{\Pi}$ is
non-empty so there are $p$ possible partitions in $\mathcal{P}$. In our
algorithm, at each stage we choose to compare $\phi\left(  S_{\Pi},S_{\bar
{\Pi}}\right)  $ over $\min(p,512)$ different partitions $\{\Pi,\bar{\Pi}\}$,
chosen at random from this set of $p$ possible partitions.

\paragraph{Block Rearrangement Algorithm (Block RA1)}

\begin{enumerate}
\item Select a random sample of $n_{sim}$ possible partitions of the columns
\thinspace$\{1,2,...,n\}$ into non-empty subsets $\{\Pi,\bar{\Pi}\}.$ Note if
$n_{sim}=2^{n-1}-1$, all partitions are considered.

\item For each of the above partitions compute $\rho_{\Pi}=\phi\left(  S_{\Pi
},S_{\bar{\Pi}}\right)  .$ Identify the partition with the largest value of
$\rho_{\Pi}.$

\item Rearrange the second block so that $S_{\bar{\Pi}}$ is countermonotonic
to the values of $S_{\Pi}.$

\item Compute the value of $\varrho(\mathbf{X})=\frac{1}{2^{n-1}-1}%
\sum_{\text{ }\Pi\in\mathcal{P}}\phi\left(  S_{\Pi},S_{\bar{\Pi}}\right)  $

\item If\footnote{This condition can be replaced by a number close to -1 such
as -0.9999.} $\varrho(\mathbf{X})>-1,$ return to step 1. Otherwise, output the
current matrix $\mathbf{X}$.
\end{enumerate}

\begin{rem}
In the selection of a candidate partition (in step 2 above), Pearson
correlation can be used in place of Spearman correlation. On the one hand, there is a
significant computational advantage because Pearson correlation between all
possible partitions is a function of the covariance matrix, which can be
computed once only. On the other hand, the effect of the RA on the Spearman correlation is very
clear as it replaces it by -1 after the algorithm is applied, whereas the
effect of Pearson correlation on the variance of the sum cannot be easily
predicted before running the RA. Using Pearson correlation  is more appropriate for large matrices.
\end{rem}

The example of matrix $C$ given in \eqref{matC} shows that the block
rearrangement algorithm is more likely to identify a dependence structure that
minimizes the variance since the standard RA may converge to a matrix
${\mathbf{X}}$ such that $\varrho(\mathbf{X})\neq-1,$ whereas the block RA
presented above ensures that the resulting matrix is such that $\varrho
(\mathbf{X})$ is $-1.$ However, the contraposive of Proposition \ref{PP1} is
not true, thus there are situations for which $\varrho\left(  \mathbf{X}%
\right)  =-1$, and thus $\phi\left(  \sum_{i\in\Pi}X_{i},\sum_{i\in\bar{\Pi}%
}X_{i}\right)  $ is minimized for every partition in two sets and the variance
is not minimized. That is, we find a local minimum for the block RA presented
above. Consider for example the matrices $A_{1}$ and $A_{2}$:
\[
A_{1}=\left(
\begin{array}
[c]{cccc}%
0.0662 & -0.9444 & 0 & -0.5842\\
0.6524 & 1.0061 & -0.0549 & 0.2495\\
0.3271 & -0.6509 & -1.3218 & -0.0833\\
1.0826 & 0.2571 & 0.9248 & -0.9263
\end{array}
\right)
\]

\[
A_{2}=\left(
\begin{array}
[c]{cccc}%
0.0662 & -0.9444 & 0 & -0.5842\\
0.6524 & -0.6509 & -0.0549 & 0.2495\\
0.3271 & 1.0061 & -1.3218 & -0.0833\\
1.0826 & 0.2571 & 0.9248 & -0.9263
\end{array}
\right)  .
\]

Applying the block RAs described above to these initial matrices $A_{1}$ and
$A_{2}$ (with a stopping rule of $\varrho(\mathbf{X})=-1$), results in convergence to
two different matrices $B_{1}$ and $B_{2}$ given by (\ref{B2}) and (\ref{B1})
with different row sums having variances 0.04346, and 0 respectively, and
multivariate dependence measure $\varrho(B_{1})=\varrho(B_{2})=-1$. For the
various possible permutations of the columns of the matrix $A_{1}$, there is a
number of possible limit matrices or local minima, with variance of the row
sums equal to $0,0.0049,0.0151,0.0217$, and $0.0435$ and over one third of the
possible starting permutations (27 of 72) lead to limits that do not minimize
the variance of the row sums. For small matrices this appears to be the rule
rather than the exception. For example, for randomly generated $4\times4$
matrices with independent $\mathcal{N}(0,1)$ distributed elements, the vast
majority (more than 80\%) appear to possess multiple local minima, in many
cases five or more as in the example above. It should not be surprising that
there may be several local minima, since this is a discrete optimization
problem, less smooth when there is a small number of rows. Moreover the
\textit{order} in which the partitions are selected may effect which local
minimum convergence is to. If the global minimum is required, then we can
begin with a number of different starting configurations, and also rely on the
randomness of the Block RA2 and see whether convergence is to a common point.
Note that this algorithm can be applied instead to a subset of the rows of
$\mathbf{X}$, but when $\varrho({\mathbf{X}})=-1$, no further improvement is
possible even on a subset of the rows.


The block RA will, in a finite amount of time, end up with a $\Sigma
$-countermonotonic structure
(\cite{puccetti2014general,lee2014multidimensional}) in which all sums over
$\Pi$ and $\bar{\Pi}$ are countermonotonic. To avoid the computationally
expensive calculation of $\varrho(\mathbf{X})$ and $\rho_{\Pi}$ for each
$\Pi,$ we have the following variation on the Block RA that we will use throughout our examples.

\paragraph{Block Rearrangement Algorithm 2 (Block RA2)}

\begin{enumerate}
\item Select a random sample of $n_{sim}=\min(512,2^{n-1}-1)$ possible partitions of the columns
\thinspace$\{1,2,...,n\}$ into non-empty subsets $\{\Pi,\bar{\Pi}\}.$ Note if
$n_{sim}=2^{n-1}-1,$ all partitions are considered.

\item For each of the above partitions, rearrange the second block so that
$S_{\bar{\Pi}}$ is countermonotonic to the values of $S_{\Pi}.$

\item If there is no improvement in var$\left(  \sum_{i}\mathbf{X}_{i}\right)
,$ output the current matrix $\mathbf{X}$, otherwise return to step 1.
\end{enumerate}

\begin{rem} Note  that the choice of $n_{sim}$  governs a
trade-off between complexity of one step of the algorithm and the number of
steps required for eventual convergence. Regardless of the value of $n_{sim},$
the algorithms converge to the same set of possible local minima but the value
of $n_{sim}$ governs the speed of that convergence. The relationship between
the speed of convergence and $n_{sim}$ is complicated since it depends on the
values in the matrix $\mathbf{X}$ and the current set of correlations $\{\phi\left(
S_{\Pi},S_{\bar{\Pi}}\right)  ;\Pi\in\mathcal{P}\}.$  Of course the
computational speed also depends on the size of the matrix, which affects the
time required to calculate the set of correlations  $\{\phi\left(  S_{\Pi
},S_{\bar{\Pi}}\right)  ;\Pi\in\mathcal{P}\}.$ The ideal choice of $n_{sim}$
 from a computational point of view in a particular problem may have to be
determined experimentally but theoretically, as mentioned above, the same set
of candidate minima result from any choice of $n_{sim}\geq1.$
\end{rem}

\subsection{Comparison of performance of the RA and Block RA}

In this section, we compare the performance of the RA and BRA in achieving the
global minimum or in approximating it. When there are three variables $(n=3)$,
the RA and the BRA are equivalent as all blocks from the Block RA correspond
to 1 column in one block and 2 columns in the other block. Therefore, there is
no reduction in variance for $n=3$. In what follows, we concentrate ourselves
to cases when $n\ge4$.

When there is a small number of columns (for example $n\leq15$), we are able
to do a block RA taking all possible partitions into two blocks ($n_{sim}%
=2^{14}-1=16,383$), with the multivariate correlation $\varrho$ computed
exactly and, on termination, equal to -1.

For small matrices (less than 10 rows and 4 columns), we can determine the
global minimum by trying every permutation of the columns.\footnote{It is also
possible to use a linear programming solver to solve for the global minimum.
It could typically handle slightly larger matrices.} We then run the RA and
the BRA to test whether they reach the global minimum, and if they do not,
then we compute  how far they are from this global minimum. Specifically, we repeat 10,000 times the
following experiment:

\begin{itemize}
\item Initialize the matrix $\mathbf{X}$ by simulating $m$ independent
Uniform[0,1] for the first column and then placing random permutations of
these same values in the remaining $n-1$ columns.

\item If $m\le10$ rows and $n\le4$ columns, permute columns $2,3,...,n-1$ in
all $\left(  m!\right)  ^{n-2}$ ways, and arranging column $n$ so that it is
countermonotonic with the sum of the other columns. Among all these configurations, find the matrix
$\mathbf{X}^{*}$ whose row sums have the global minimum variance $V^{*}.$

\item Apply the standard RA to $\mathbf{X}$ to obtain a local minimum
$\mathbf{X}_{ra}$ in which all columns are countermonotonic to the sum of the
others. Compute the variance $V_{ra}$ of the row sums of $\mathbf{X}_{ra}$.

\item Apply the block RA, BRA2, to $\mathbf{X}_{ra}$ to obtain the matrix
$\mathbf{X}_{bra}$ and the variance $V_{bra}$ of the row sums of
$\mathbf{X}_{bra}$.
\end{itemize}

\paragraph{How much does the BRA improve upon the RA?\newline}

In order to compare the RA and the BRA, we compute the average value for
$V_{ra}$ and for ${V_{bra}}$. The results are reported in Table \ref{Tcomp}.
\begin{table}[!htbp]
\caption{Average variance for the RA and for the Block RA. Both averages are
estimated with 10,000 experiments as described above for different values of
$n$ and $m$. All digits reported in the table are significant. }%
\label{Tcomp}%
\begin{center}%
\begin{tabular}
[c]{c||c|c||c|c||c|c|}
& \multicolumn{2}{c||}{$n=4$} & \multicolumn{2}{|c||}{$n=7$} &
\multicolumn{2}{|c|}{$n=10$}\\
average of & $\quad{V_{ra}}\quad$ & $\quad{V_{bra}}\quad$ & $\quad{V_{ra}%
}\quad$ & $\quad{V_{bra}}\quad$ & $\quad{V_{ra}}\quad$ & $\quad{V_{bra}}\quad
$\\\hline
$m=10$ & 0.001 & 0.0006 & 0.0004 & 1.1$\times10^{-5}$ & 0.00018 &
1.8$\times10^{-7}$\\
$m=100$ & 1.2 $\times10^{-5}$ & 5.5 $\times10^{-6}$ & 3.4$\times10^{-6}$ &
8$\times10^{-8}$ & 1.6$\times10^{-6}$ & 1.3$\times10^{-9}$\\
$m=1,000$ & 1.2$\times10^{-7}$ & 5.5$\times10^{-8}$ & 3.2$\times10^{-8}$ &
7.6$\times10^{-10}$ & 1.6$\times10^{-8}$ & 1.2$\times10^{-11}$\\\hline
\end{tabular}
\end{center}
\end{table}

We make the following observation on Table \ref{Tcomp}. The larger
the number of variables $n$ or the number of discretization steps $m$, the larger
the improvement of the Block RA over the RA. We have performed other
experiments with other distributions and we obtain similar results.


\paragraph{Convergence of the RA and BRA algorithms to the global minimum
variance\newline}

Both the RA and the Block RA2 terminate because they are based on the variance
of the row sums that decreases strictly at each step and is bounded from below
by 0, and because there is a finite number of permutations, hence a finite
number of values of this quantity. However, as we have shown, it is possible
to end up at a local minimum of the variance instead of the global minimum.

In Figure \ref{F00}, we plot the percentage of cases in which $V_{ra},$ resp.
$V_{bra}$ is within a given tolerance\footnote{$10^{-6}$ in this case.} of
$V^{*}.$ It shows that this percentage decreases quickly to 0 as  $m$ increases. Table \ref{T1b} reports  the averages of the difference $V_{ra}-V^{*}$ and $V_{bra}-V^{*}.$ We find
that the Block RA outperforms the RA by several orders of magnitude. 

\begin{table}[!htbp]
\caption{Average distance from the minimum for the RA and the BRA for $n=4$
variables.}%
\label{T1b}%
\begin{center}%
\begin{tabular}
[c]{c|c|c|c|c}
& $m=4$ & $m=5$ & $m=6$ & $m=7$\\\hline
RA & 0.0020 & 0.0015 & 0.0026 & 0.0015\\\hline
BRA & 0.0001 & 0.0002 & 0.0003 & 0.0003\\
\end{tabular}
\end{center}
\end{table}

\begin{figure}[!h]
\begin{center}
\includegraphics*[width=10cm,height=7cm]{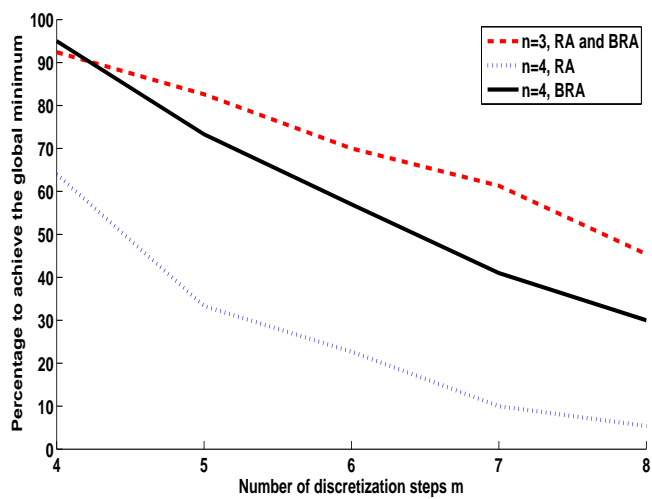}
\end{center}
\caption{Percentage of cases that the minimum from the algorithms RA or BRA
get to the global minimum (within $10^{-6}$)}%
\label{F00}%
\end{figure}

\begin{rem}
The comparison in Table \ref{T1b} is necessarily done with very small matrices as the
global minimum $V^{*}$ is computed by computing the variance in all possible
permutations of the matrix. For larger matrices, such a technique cannot be
used. In fact, there are very few cases for which we know the value of the
minimum. One option is to use the result of \cite{Haus} that gives the
minimum variance in the case of a matrix $m$ by $n$ that contains in each
column the integers $1,2,...,m$. But this is a very specific case. In this
case (at the minimum variance matrix), the mean of the sum is $\mu:=n\frac{1+2+...+m}{m}$ and the sum takes two
values $M=\lfloor\mu\rfloor$ with probability $q:=\mu-\lfloor\mu\rfloor$ and
$M+1$ with probability $(1-q)$ so that the minimum global variance can be
computed explicitly. We are then able to check the percentage of the time the
RA, respectively the BRA, achieves the global minimum by starting from a
randomized matrix (where each column has been randomly permuted). We obtain
similar conclusions as in Figure \ref{F00} and Table \ref{T1b}, namely the
percentage of the cases in which one achieves  the global minimum decreases with $n$ and
$m$ and the BRA is closer to the global minimum by several order of magnitude.
\end{rem}

In order to assess the convergence of the algorithm with larger matrices in a
more general setting, we propose to generate matrices that all have constant
row sums so that the variance of the row sums is zero. We then randomly
rearrange the values in each column, and then the RA and block RA2 can be
applied to the matrix to see to what extent the minimum variance is achieved.
For instance, we can generate a matrix of $\mathcal{N}(0,1)$ random variables
with constant row sums as follows. First, generate independent $\mathcal{N}%
(0,1)$ random variables, then subtract the row mean from each row so that the
sum is now $0.$ Lastly, multiply by the factor $\frac{m}{m-1}$ in order to
return the marginals to $\mathcal{N}(0,1).$

Applying the Block RA2 with this matrix of $\mathcal{N}(0,1)$ variables, we
obtain the results in Table \ref{T3b}. Neither the RA nor the BRA guaranteed achieving the
minimum possible variance in the simulations because of the complexity of the
problem. Note, however that the BRA is between 10 and 40 times closer to the
optimum than is the RA for $n\geq5.$ These two conclusions are consistent with
the previous examples.

\begin{table}[!htbp]
\caption{Average distance from the minimum for the RA and the BRA with 10,000
simulations with $m=10$.}%
\label{T3b}%
\begin{center}%
\begin{tabular}
[c]{c|c|c|c|c|c}
& $n=4$ & $n=5$ & $n=6$ & $n=7$ & $n=8$\\\hline
RA & 0.02 & 0.001 & 0.007 & 0.005 & 0.004\\\hline
BRA & 0.005 & 0.0009 & 0.0003 & 0.0001 & 0.00009\\
\end{tabular}
\end{center}
\end{table}

The RA and Block RA have been developed to minimize the expectation of a
convex function of the sum of dependent random variables with given marginal
distributions. As shown in \cite{Haus}, checking the complete mixability
condition is a NP-complete problem (even in the case of 3 variables only), and
therefore there exists no algorithm with polynomial complexity that converges
to the global minimum with certainty. Neither the RA nor the Block RA
guarantees convergence to the global minimum. Furthermore, our counterexample
\eqref{B2} and \eqref{B1} in Section \ref{S1} shows that the block RA may end
up in a strict local minimum with a positive variance while the global minimum
for the variance is equal to 0. Nevertheless the Block RA seems to approximate the global minimum to a reasonable degree of precision for large matrices.

\section{Application to finding the dependence to get a target distribution
for the sum \label{S3}}

As discussed in the introduction, the Rearrangement Algorithm has been widely used in finance and risk management. In this section, we discuss a new application as to infer the joint distribution among variables for which the distribution of the sum, the difference, or a weighted sum is known. We first illustrate the methodology with sums of normal or uniform variables. Next, we discuss a real-world application with the example of spread options.

\subsection{Indistinguishability}

We base our analysis on two goodness of fit test statistics. The first one is the
Kolmogorov-Smirnov (KS) test. It is fully non-parametric and applies to all
target distributions. The second one is less well-known but based on a more
appropriate measure of distance in our context, the $L^{2}$-Wasserstein
distance measure. The KS test is based on the following results. Suppose
$F_{m}$ is the empirical c.d.f. from a sample of size $m$ with true
distribution $F$. Define $D_{m}=\sup_{x}|F_{m}(x)-F(x)|$, then the asymptotic distribution of $\sqrt{m} D_m$  is well-known.\footnote{$\lim_{m\rightarrow+\infty}P(\sqrt{m}D_{m}\leq t)=H(t):=1-2\sum_{k=1}^{\infty
}(-1)^{k-1}e^{-2k^{2}t}.$}
We may use this asymptotic result to determine the median of the distribution for large $m$ or use simulations to approximate this value for finite $m$.  For example, when $m=10^6$, using simulations, we obtain that 
the median of $D_{m},$  $med_F(D_m)$, is approximately equal to 8.2$\times 10^{-4}$
 so
that any distribution within a region $F(x)\pm 8.2 \times 10^{-4}$ will fall
in a pointwise 50\% confidence interval around $F.$ This is a very strong
result as it implies for example that it falls in all standard (e.g., 95\%,
99\%) confidence intervals.

If $G(x)$ falls in such an interval based
on a sample of $m=10^{6}$ observations, it is usually indistinguishable from
the target cdf $F$. So for the purpose of this paper, we define:

\begin{definition}
\label{dKS} $G(x)$ is empirically KS-indistinguishable from $F(x)$ with a sample size of
$m=10^{6},$ if
\begin{equation}
\sup_{x}|G(x)-F(x)|\leq med_F(D_m)\label{KSdist}%
\end{equation}
\end{definition}

The KS test and therefore Definition \ref{dKS} applies to any cdf $F$. Of
course, other test statistics might also be used with empirical data to
determine the fit of the normal distribution. Observe also that the KS test is
based on the distance between the cdfs, using the distance $D_{m}$ defined
above. But the test statistic most consistent with the rearrangement algorithm
 is the $L^{2}$-Wasserstein metric, which measures the squared $L^{2}$ distance
between the quantile functions (see for example \cite{krauczi2009study}),
$\int_0^1|G^{-1}(u)-F^{-1}(u)|^{2}du$.

Let us define the following distance from an empirical quantile function
$F_{m}^{-1}(u)$ to the distribution $F$ (related to the $L^{2}$-Wasserstein squared distance)
\begin{equation}
T_{m}=\int_{0}^{1}|F_{m}^{-1}(u)-F^{-1}(u)|^{2}%
du\label{L2W-diist}
\end{equation}
Analogous to
Definition \ref{dKS}, we define 

\begin{definition}\label{sKS}
$G(x)$ is empirically $L^{2}-W$-indistinguishable from $F(x)$ if $\int_{0}%
^{1}|G^{-1}(u)-F^{-1}(u)|^{2}du\leq med_{F}(T_{m})$ \ for $m=10^{6}.$
\end{definition}
The asymptotic distribution of $T_m$ depends on $F$ (unlike to the KS distance that was discussed above).
The asymptotic distribution of $T_m$ is known for various $F$ (see for
example \cite{del1999tests}). However, it is a functional of a Brownian bridge and
so for finite $m$ we again determine the median from a simulation.

When the distribution $F$ is the standard normal distribution $N(0,1)$, then $med_{F}(T_{m})$ is approximately $3.5\times10^{-6}$ and when $F$ is $\mathcal{U}[-1,1]$, then it is approximately $4.7\times10^{-7}.$ 
Combining these two test statistics, we thus define the notion of empirical indistinguishability.

\begin{definition}
$G(x)$ is empirically indistinguishable from $F(x)$ if it is both empirically
$L^{2}-W$-indistinguishable from $F(x)$ and empirically $KS$%
-indistinguishable from $F(x)$.
\end{definition}

\subsection{Sum of two or more uniform distributions}

Perhaps surprisingly, there is a copula such that the sum of $n\geq2$ uniform
random variables is empirically indistinguishable from a normal distribution.
For convenience, we choose expected values equal to 0 and $X_{i}$ are
uniformly distributed over $[-a,a]$, i.e. $X_{i}\sim\mathcal{U}%
[-a,a],i=1,2,...,n$. We want to show there is a dependence structure such that
the sum of a given number of uniform random variables on $[-a,a]$ is close to
Normal $\mathcal{N}(0,1)$ distributed. In other words, we seek a copula for
random variables $X_{1},X_{2},...,X_{n}$ where $X_{i}\sim\mathcal{U}%
[-a,a],i=1,...,n$ and $X_{n+1}=-Z$ is $\mathcal{N}(0,1)$ such that
$\text{var}(S)=\text{var}(X_{1}+X_{2}+...+X_{n+1})$ is minimized. \ 

\paragraph{Block RA with $\mathcal{U}[-a,a]$ to achieve a $\mathcal{N}(0,1)$}

\begin{enumerate}
\item Start with an initial value of $a=1.5.$

\item Run the block RA. \ Periodically, at each step 1 of the block RA,
replace $a$ by a constant chosen such that $\text{var}(\sum_{i=1}^{n}%
X_{i})=1.$

\item Terminate the block RA when both  var$(S)$ and the value of $a$ fail to
change by a given tolerance. \ 
\end{enumerate}

This algorithm permits finding a copula such that the sum of
$n$ uniform $\mathcal{U}[-a,a]$ is indistinguishable from a normal random
variable. We can change the target distribution to a variety of distributions
and still get near equality up to a change in the location and scale of the target.

\begin{proposition}
For each $n\geq2,$ there exist a copula and a value of $a>0$  such that the
sum of $n$ dependent $\mathcal{U}[-a,a]$ with that copula is empirically indistinguishable from the
$\mathcal{N}(0,1)$.
\end{proposition}

This theoretical result is not surprising. Ruodu Wang pointed out to us that
all unimodal-symmetric distributions supported in $[-2a,2a]$ can be
represented as the distribution of the sum of 2 variables uniformly
distributed on $[-a,a]$. The existence of such dependence structure between
two uniform variables can be proved using arguments of joint mixability
(\cite{wang2015joint}). The proof of this proposition requires only that we
choose $a$ large enough that the KS and the $L^2-$Wasserstein differences
between the standard normal distribution and the normal distribution constrained to lie
in the interval $[-2a,2a]$ is small, say less than $10^{-6}$, and then
represent this conditional normal distribution with the sum of two random
variables uniformly distributed on $[-a,a]$. Our approach makes it possible to
construct the explicit dependence between the two uniform variables
numerically and identify a suitable value of $a.$ 

We will verify it by a numerical evaluation of the integrals in KS distance \eqref{KSdist} and $L^{2}$-W distance \eqref{L2W-diist} with 50,000 steps.
Table \ref{T1} confirms the result for $n=2,3,4$. The critical value are respectively given by $med(D_m)=8.2 \times 10^{-4}$ and $med(T_m)=3.5 \times 10^{-6}$ 
\begin{table}[!h]
\caption{Sums of Uniform $\mathcal{U}[-a,a]$ and target cdf is a normal
$\mathcal{N}(0,1)$. We report the values of the KS distance in the second
column and the $L^{2}-W$ distance in the third column. The $L^{2}-W$ distance
is the variance of $X_{1}+X_{2}+...+X_{n}-Z$ where $Z$ has cdf $\mathcal{N}%
(0,1)$ and $X_{i}\sim\mathcal{U}[-a,a]$.}%
\label{T1}%
\begin{center}%
\begin{tabular}
[c]{c||c|c|c}%
n & KS distance & $L^{2}-W$ distance & a\\\hline
2 & $5.5 \times10^{-5} $ & $9.3\times10^{-7}$ & 2.08\\
3 & $3.2 \times10^{-5} $ & $4.9\times10^{-7}$ & 1.38\\
4 & $2.2 \times10^{-5} $ & $3.5\times10^{-10}$ & 1.31
\end{tabular}
\end{center}
\end{table}

We used a numerical evaluation of the distances (\ref{KSdist}) and
(\ref{L2W-diist}) using a grid of 50,000 points. For any $n>4,$ if $n$ is odd, we can build the copula for the first
three columns and then add countermonotonic  pairs of uniform random variables
$X_{i}=U_{i}$, $X_{i+1}=-U_{i}$ etc., where $U_{i}$ are $\mathcal{U}[-a,a]$
and independent of the first three columns. Similarly, we treat the case $n>4$
when $n$ is even. Indeed we obtain Kolmogorov Smirnov distances well within
the above-mentioned bound of $0.00082.$ Moreover, the observed value of
$T_{m}$ is again well within the 50\% confidence interval based on the
statistic $T_{m}.$\hfill$\Box$

The joint density of this copula is obtained in Panel A of Figure \ref{F1}
using a nonparametric density estimator for 10,000 data values. In general,
 the standard RA  leads to a bivariate density that is significantly less smooth than the one
obtained by the Block RA.

\subsection{Sum of two or more normal distributions}

If we reverse the roles of these two distributions, we can show that the sum
of two dependent $\mathcal{N}(0,\sigma^{2})$ random variables is
KS-empirically indistinguishable from a $\mathcal{U}[-1,1]$ random variable
with joint density of the copula displayed in Panel B of Figure \ref{F1}
below. In this case the target distribution is the $\mathcal{U}[-1,1]$ and the
asymptotic results for the $L^{2}$-Wasserstein test are complex so we
obtained the critical value $med_{F}(T_{m})\approx 4.7\times 10^{-7}$ from simulations with
$m=10^{6}.$

\begin{figure}[!htbp]
\begin{center}%
\begin{tabular}
[c]{cc}%
\includegraphics[
height=2.5in,
width=3.5in
]{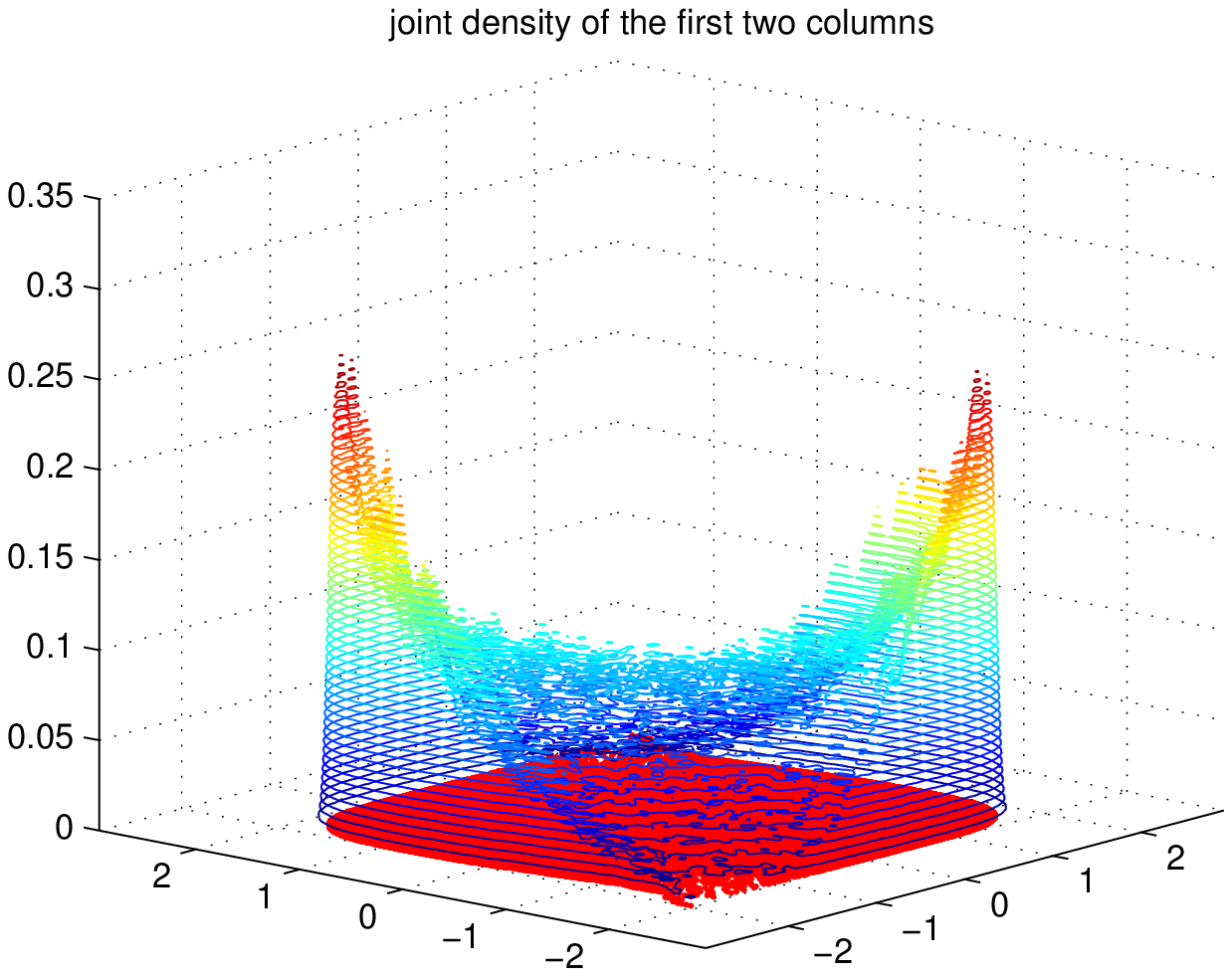} & \includegraphics[
height=2.5in,
width=3.5in
]{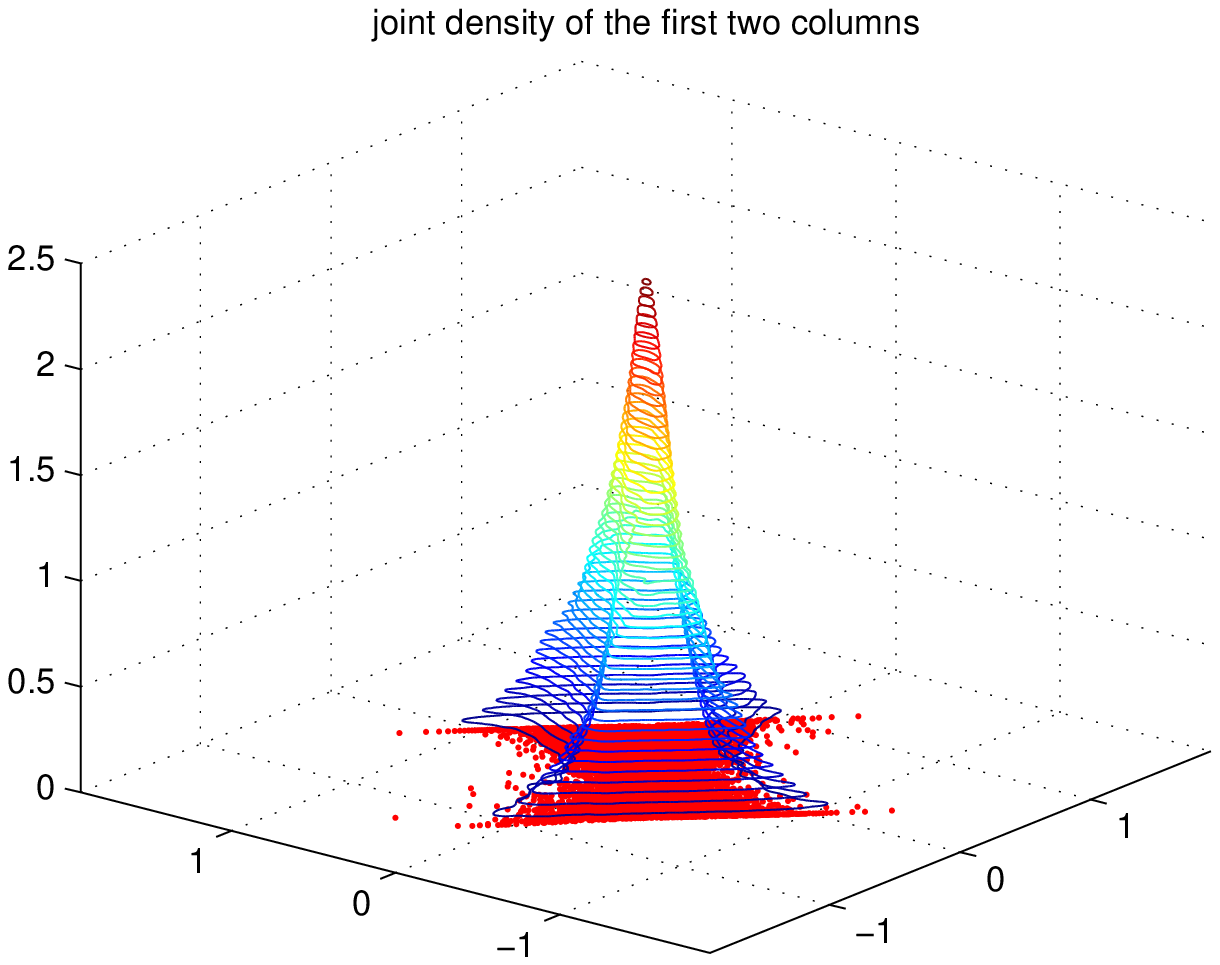}\\
Panel A & Panel B\\
&
\end{tabular}
\end{center}
\caption{Panel A: Joint density of two $\mathcal{U}[-2.056,2.056]$ random
variables whose sum is indistinguishable from $\mathcal{N}(0,1).$ Panel B: Joint density of the two
marginally normal random variables $\mathcal{N}(0,0.3363^{2})$ whose sum is
indistinguishable from $\mathcal{U}[-1,1].$}%
\label{F1}%
\end{figure}


\begin{proposition}
For each $n\geq2,$ there exist a copula and a value of $\sigma^{2}>0$ such that the sum of $n$ dependent $\mathcal{N}(0,\sigma^2)$ with that copula is
empirically indistinguishable from $\mathcal{U}[-1,1]$.
\end{proposition}

Table \ref{T2} hereafter provides the result for $n=2,$ $3$ and $4$. However,
for any $n>3,$ we can build the copula for the first two or three columns and
then add countermonotonic  pairs of random variables $X_{4}=Z_{4},X_{5}=-Z_{4}$ etc.
where $Z_{i}$ are independent $\mathcal{N}(0,\sigma^{2})$ independent of the
first two or three columns. Once again we use a numerical evaluation of the
distances (\ref{KSdist}) and (\ref{L2W-diist}) using a grid of 50,000 points. This verifies that such a copula exists
for any $n\geqslant2$.

\begin{table}[!h]
\caption{Sums of Normal $\mathcal{N}(0,\sigma^{2})$ and target cdf is a
Uniform over $[-1,1]$. We report the values of the KS distance in the second
column and the $L^{2}-W$ distance in the third column. This is the variance of
$X_{1}+X_{2}+...+X_{n}-Z$ where $Z$ has cdf $\mathcal{U}[-1,1]$ and $X_{i}%
\sim\mathcal{N}(0,\sigma^{2})$.}%
\label{T2}%
\begin{center}%
\begin{tabular}
[c]{c||c|c|c}%
n & KS distance & $L^{2}-W$ distance & $\sigma$\\\hline
2 & 4.8$\times10^{-5}$ & 2.3$\times10^{-10}$ & 0.3363\\
3 & 2.7$\times10^{-5}$ & 4$\times10^{-11}$ & 0.4\\
4 & 1.9$\times10^{-5}$ & 7$\times10^{-12}$ & 0.45
\end{tabular}
\end{center}
\end{table}

\subsection{Application in Finance}

The above examples using normal and uniform variables may suggest that the methodology has limited practical implications. This is not correct as the methodology can be  very  useful in finance to infer the dependence among assets in the risk-neutral world (i.e., using option prices as sole available information). We briefly outline the methodology and show how it can be used to choose a pricing model for spread options that is consistent with option prices on each asset and on the spread (difference). It is particularly interesting in fitting the bivariate distribution between the gas price and the electricity price given that spark spread options are actively traded. Spark spread options are options on the spread between natural gas and electric power as $S_t=P_t-hG_t$ where $P_t$ and $G_t$ denote futures prices of power and gas, and  $h$ is the heat rate or efficiency ratio of a typical gas fired power plant (see \cite{alexander2004bivariate},  \cite{carmona2003pricing},  \cite{rosenberg2000nonparametric}). Let $T$ be the maturity of all options under consideration. From option prices on an asset with a large number of strikes it is possible to infer the marginal distribution of this asset (\cite{ait2000nonparametric}, \cite{breeden1978prices},
\cite{bondarenko2003estimation}). Assuming that prices of spread options are available at the same time as  options written on gas and electricity, it is possible to infer the marginal distributions of gas and electricity returns and of the spread. Our methodology can then be used to infer the dependence between gas and electricity returns as follows
\begin{enumerate}
	\item Use options on gas prices, on electricity prices and on the spark spread to  derive the distribution function  $F_P$ of $P_T$, $F_G$ of $hG_T$, $F_S$ of $S_T$ respectively (e.g., following the methodology of \cite{ait2000nonparametric}, \cite{breeden1978prices},
\cite{bondarenko2003estimation}).
\item For a given maturity, apply the block RA on a matrix with 3 columns and $m$ rows (where $m$ is the number of discretization  steps). Each column contains a discretized distribution:
	\begin{itemize}
		\item In the first column
		$$F_{P}^{-1}\left(\frac{i}{m+1}\right)\quad \quad i=1,2,...,m$$ 
		\item In the second column
		$$-F_{G}^{-1}\left(\frac{i}{m+1}\right)\quad \quad i=1,2,...,m$$ 
		\item In the third column
		$$-F_S^{-1}\left(\frac{i}{m+1}\right), \quad \quad i=1,2,...,m$$
		\item Apply the Block RA on the full matrix
	\end{itemize}
	
\underline{Output:} Extract the first two columns to describe a discrete copula that is consistent with the information on the marginal distributions of $P_T$, $G_T$, and of the spread $S_T$. 
\end{enumerate}

By repeating the above experiments sufficiently many times, one can describe models that are consistent with the information of the margins $F_P$, $F_G$ and $F_S$.

\section{Alternative approach: MCMC Block RA \label{SMCMC}}

As mentioned earlier, the problem of converging to the minimum of a convex function of the sum is NP-complete. It is thus not possible to find a deterministic algorithm that converges to the global minimum in polynomial time. We end the paper with an alternative direction that relies on a  stochastic algorithm to achieve the convergence to the global minimum asymptotically in polynomial time (Theorem \ref{t3}). 

The RA and the Block RA converge to a possible solution for the global minimum
of the expectation of a convex function of the sum (Proposition \ref{PP1}).
However, when there are more than 3 variables involved, the dependence
structure that achieves the global minimum variance does not necessarily
minimize other convex functions of the row sums. In this section, we develop a
stochastic algorithm that is able to identify the global minimum in finite time.

Consider the matrix $\mathbf{X}=[\mathbf{X}_{1},\mathbf{X}_{2},...,\mathbf{X}%
_{n-1},\mathbf{X}_{n}].$ For a set of columns $\Pi,$ we denote by
$S_{\bullet\Pi}=S_{\bullet\Pi}(\mathbf{X})$ the vector of sums
\[
S_{\bullet\Pi}:=\left(  \sum_{k\in\Pi}X_{ik}\right)  _{i=1,2,...,m}
\]
and by $\mathbf{X}_{\bullet\Pi}$ the submatrix $X_{ik}$, $i=1,...,m$, $j\in\Pi.$
Assume, without loss of generality, that the column sums of $\mathbf{X}$ are
all zero. For simplicity, consider the partition $\Pi=\{1,2,...,k\}$ and $\bar{\Pi}%
=\{k+1,...,n\}.$ We consider operations, which rearrange the rows in $\bar
{\Pi}$ while keeping those in $\Pi$ unchanged. The mean of the vector
$S_{\bullet\Pi}+S_{\bullet\bar{\Pi}}$ is unchanged. For a positive convex
function $f(s)$ of these row sums,
\[
f(S_{\bullet\Pi}+S_{\bullet\bar{\Pi}})\geq f(S_{\bullet\Pi}+S_{\bullet\bar
{\Pi}}^{a})
\]
where $S_{\bullet\bar{\Pi}}^{a}$ consists of the same components as
$S_{\bullet\bar{\Pi}}$ but arranged to be countermonotonic to $S_{\bullet\Pi
}.$ An operation, which rearranges the rows in $\bar{\Pi}$
countermonotonically, while keeping those in $\Pi$ unchanged results in a
reduction in a convex loss function. This choice is the basis of the block RA.

We now design a stochastic algorithm to determine local and global minima of
$f(S_{\bullet\bullet}(\mathbf{X})),$ where $S_{\bullet\bullet}(\mathbf{X})$ denotes the $m$ sums over all columns. In particular, we define
\[
\ell(\mathbf{X})=\frac{1}{f(S_{\bullet\bullet}(\mathbf{X}))}%
\]
and construct a Markov Chain designed to find the maxima of $\ell
(\mathbf{X}).$ Any other distribution $\ell(\mathbf{X})$ whose probabilities
are decreasing functions of $f(S_{\bullet\bullet}(\mathbf{X)})$ such as
$\exp(-Tf(S_{\bullet\bullet}(\mathbf{X)}))$ for some $T>0$ would also suffice. We choose a
random partition $\Pi$ uniform over the $2^{n-1}-1$ possible partitions. We
then propose a random rearrangement of the rows of $X_{ik},k\in\bar{\Pi}$
designed so that after the rearrangement, $S_{\bullet\Pi},S_{\bullet\bar{\Pi}%
}$ will \emph{tend to be} countermonotonic. We then \textquotedblleft accept"
the move to this new matrix $\mathbf{X}^{\prime}$, say, with probability
\[
\min\left(  1,\frac{\ell(\mathbf{X}^{\prime})}{\ell(\mathbf{X})}\right)  .
\]
Note that larger values of $\frac{\ell(\mathbf{X}^{\prime})}{\ell(\mathbf{X}%
)}$ tend to lead to acceptance of the move, and smaller values tend to result
in remaining at $\mathbf{X.}$ We arrange that for a given partition $\Pi$ the proposal depends only on the
matrix $\mathbf{X}_{\bullet\Pi}$. If $\ell(\mathbf{X})<\infty$ for all
 $\mathbf{X}$, this algorithm results in a finite state ergodic Markov Chain,
which converges to a stationary distribution with positive probability on all
possible states of the chain so that states with a very small value of
$f(S_{\bullet\bullet}(\mathbf{X)})$ appear with higher frequency.

We wish to select a random permutation $s^{\ast}$ of the rows of $\mathbf{X}%
_{\bullet\bar{\Pi}}$, which depends only on $S_{\bullet\Pi}$ in such a way
that, after rearrangment, $S_{\bullet\Pi},S_{\bullet\bar{\Pi}}$ tend to be
countermonotonic. To do so, we choose $S_{\bullet\bar{\Pi}}$ to be ranked
identically to independent observations from a location family of
distributions $Y_{i}-S_{i\Pi}$ where $Y_{i}\sim g(y)$. We might choose $g(y)$
to be normally distributed with mean $0$ and variance $\sigma^{2}$ or any
other location family of distributions. We used $g(y)$ following the Gumbel
extreme value distribution\footnote{A similar family of distributions
$g(z)=\frac{1}{\Gamma(r)}e^{-rz}\exp(-e^{-z})$ is obtained as the logarithm of
a Gamma distributed random variable and provides a random permutation from the
well-known Gamma ranking family (\cite{Stern}).}
\begin{equation}
g(z)=re^{-rz}\exp(-e^{-rz}). \label{gumbel}%
\end{equation}
When the scale parameter $1/r$ of this location family approaches $0,$ this
ranking approaches a countermonotonic one and the algorithm approaches the
Block RA. An illustration of the $m$ densities from which we simulate
independently is represented in Figure \ref{locfig}.

\begin{figure}[!htbp]
\begin{center}
\includegraphics[
height=2.1361in,
width=5.4284in
]{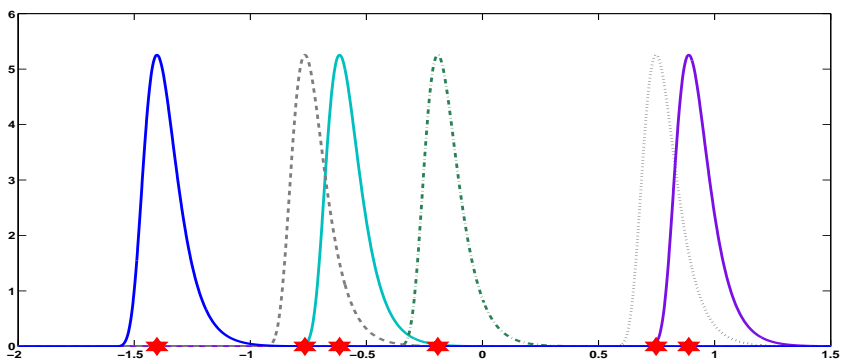}
\end{center}
\caption{Example with $m=6$ and the 6 location families from which we generate
independent random numbers corresponding to the 6 observations of
$S_{\bullet\Pi},$ these values are marked with ``*.'' }%
\label{locfig}%
\end{figure}

\textbf{Algorithm:} Repeat $n_{sim}$ times, with initial matrix
$\mathbf{X}$

\begin{enumerate}
\item Propose $\Pi.$ Determine $S_{\bullet\Pi},$ and $S_{\bullet\bar{\Pi}}$.
For independent random variables $Y_{i}$, $i=1,2,...m$ drawn from
(\ref{gumbel}) generate a random permutation $s^{\ast}(\bar{\Pi})$ by ranking
the observations $\mathbf{Y}-S_{\bullet\Pi}$. We then reorder the rows of
$\mathbf{X}_{\bullet\bar{\Pi}}$ using this permutation to obtain a proposal
matrix with $\mathbf{X}_{\bullet\bar{\Pi}}^{\prime}=\mathbf{X}_{\bullet
s^{\ast}(\bar{\Pi})}$ and leaving the columns in $\mathbf{X}_{\bullet\Pi}$
unchanged. The distribution of the permutation $s^{\ast}$ depends only on
$\mathbf{X}_{\bullet\Pi}$.

\item Accept the proposed rearrangement of the rows of $\mathbf{X}%
_{\bullet\bar{\Pi}}$ with probability proportional to $\min(1,\frac
{\ell(\mathbf{X})}{\ell(\mathbf{X}^{\prime})}),$ otherwise do not rearrange.

\item Record the states of the system $\mathbf{X}$ having small values for
$f(\mathbf{X})$ and their frequencies.
\end{enumerate}

We wish to identify the stationary distribution $\mu(\mathbf{X})$ of this
Markov Chain. Suppose $\mathbf{Y}$ is a matrix identical to $\mathbf{X}$
on the columns $\Pi$ and with columns $\bar{\Pi}$ a permutation of those
of $\mathbf{X}$, i.e. $\mathbf{Y}_{\bullet\bar{\Pi}}=\mathbf{X}_{\bullet s^{\ast}%
(\bar{\Pi})}.$ Then the transition probability matrix is defined by:
\[
P_{\mathbf{X,Y}}=\frac{1}{2^{n-1}-1}q(\mathbf{Y}|\mathbf{X}_{\bullet\Pi}%
)\min\left(1,\frac{\ell(\mathbf{X})}{\ell(\mathbf{Y})}\right).
\]
Here, $q(\mathbf{Y}|\mathbf{X}_{\bullet\Pi})$ is the probability of proposing the
permutation $s^{\ast}(\bar{\Pi})$ based on the row sums $S_{\bullet\Pi}$.
 The equilibrium distribution $\mu(\mathbf{X})$ must satisfy
\begin{align}
\sum_{\mathbf{X}}\mu(\mathbf{X})P_{\mathbf{X,Y}}  &  =\mu(\mathbf{Y}),\text{
and } \sum_{\mathbf{X}}\mu(\mathbf{X})=1.\label{7.5}
\end{align}
Although it may be difficult in general to solve this system of equations,
provided $0<\ell(\mathbf{X})<\infty,$ for all $\mathbf{X,}$ this is a finite
state irreducible positive recurrent (ergodic) Markov Chain. Therefore the
stationary distribution is such that every state has positive probability (see
Theorem, page 393, \cite{Feller}). Thus it guarantees that every state is visited
in finite time, and that the expected time before the chain visits the global
minimum is finite. If there is a matrix $\mathbf{X}$ such that $\ell(\mathbf{X})=0$, then the chain is absorbed and the algorithm terminates at this optimum.

This algorithm offers a compromise between rapid initial convergence and a
guarantee that the global minimum variance will eventually be achieved.
Depending on the choice of scale parameter, it offers a rapid convergence to a
region in which the objective function is small, followed by fluctuations
around the local minima of the function. Since every point $\mathbf{X}$ in the
sample space of all possible column permutations is visited with frequency
proportional to $\mu(\mathbf{X})$ we are guaranteed that the global minimum
will be reached in a finite amount of time. Indeed the stationary
probabilities $\mu(\mathbf{X})$ represent the reciprocals of the mean
recurrence time to this state.

\begin{theo}
\label{t3} The above algorithm generates a Markov Chain on the state space of
matrices $(\mathbf{X}_{n})_{n\in\mathbb{N}}$, which converges to its stationary distribution
$\mu(\mathbf{X})$ (see \eqref{7.5}). The probability that the global optimum $\mathbf{X}_{\min}$
is not found after $N$ simulations is $o(q^{N})$ for some $q<1.$
\end{theo}
\begin{proof}
The proof is a consequence of well-known results concerning the convergence of a finite state ergodic Markov Chain. For the geometric rate of convergence to the stationary distribution, see for example \cite{Cinlar}.
\end{proof}

We run this algorithm using as starting matrix $B_{1},$ the matrix discussed
earlier, for which for all $7$ possible partitions $\Pi,\bar{\Pi}$ we have
that $S_{\Pi},S_{\bar{\Pi}}$ are countermonotonic so that $\phi\left(
\sum_{i\in\Pi}\mathbf{X}_{i},\sum_{i\in\bar{\Pi}}\mathbf{X}_{i}\right)  =-1.$
This is a local optimum for the Block RA. The variance of the row sums is
0.04346.
\begin{equation}
B_{1}=\left(
\begin{array}
[c]{cccc}%
0.0662 & 0.2571 & 0 & -0.5821\\
0.3271 & 1.0061 & -1.3218 & -0.0833\\
0.6524 & -0.6509 & -0.0549 & 0.2495\\
1.0826 & -0.9444 & 0.9248 & -0.9263
\end{array}
\right)  . \label{B2b}%
\end{equation}

%

\begin{figure}
[!htbp]\label{F0}
\begin{center}
\includegraphics[
height=2.1361in,
width=5.4284in
]%
{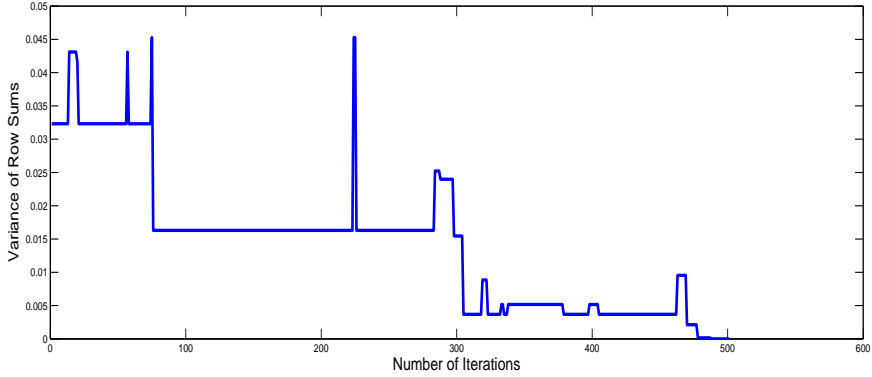}%
\end{center}
\caption{Trajectory of the above algorithm for the
initial matrix $B_1$}
\end{figure}

Figure \ref{F0} illustrates the trajectory of the above algorithm for this
initial matrix. Note that it successfully climbs out of local valleys and in
less than 500 steps is able to find the matrix corresponding to the global
minimum variance of 0. Of course the number of steps required to find this
optimum is, in general, random but if we were to use enumeration, we would
require evaluating the rows sums over a sample space of $(4!)^{3}=13824$
different matrices.

The preceding example is somewhat atypical of the performance of the algorithm
because the minimum variance is $0$ and eventually this Markov Chain is
absorbed in this state. When the minimum is strictly positive, the chain tends
to fluctuate around its equilibrium distribution described by Theorem \ref{t3}
above. 

For example, suppose we begin with the matrix $\mathbf{X}$ below, which was
obtained by generating the first column as ordered $\mathcal{U}[0,1]$ variables and the
second and third columns are random permutations of the first.%

\[
\mathbf{X}=\left(
\begin{array}
[c]{ccc}%
0.0074 & 0.8657 & 0.8574\\
0.2957 & 0.2957 & 0.3569\\
0.3569 & 0.6067 & 0.6067\\
0.4638 & 0.8574 & 0.4850\\
0.4850 & 0.0074 & 0.2957\\
0.6067 & 0.4638 & 0.8657\\
0.8574 & 0.4850 & 0.4638\\
0.8657 & 0.3569 & 0.0074
\end{array}
\right)
\]
In this case, the minimizing matrix is
\[
{\mathbf{X}}_{\min}=\left(
\begin{array}
[c]{ccc}%
0.0074 & 0.8657 & 0.6067\\
0.2957 & 0.8574 & 0.3569\\
0.3569 & 0.2957 & 0.8574\\
0.4638 & 0.4638 & 0.4850\\
0.4850 & 0.4850 & 0.4638\\
0.6067 & 0.0074 & 0.8657\\
0.8574 & 0.3569 & 0.2957\\
0.8657 & 0.6067 & 0.0074
\end{array}
\right)
\]
with variance of the row sums equal to 0.0012.

The trajectory in Figure \ref{Ffin} clearly
shows the fluctuations around a stationary distribution in the variance of the
row sums over these 5,000 iterations.%

\begin{figure}
[!hptb]\label{Ffin}
\begin{center}
\includegraphics[
height=2.1283in,
width=5.38in
]%
{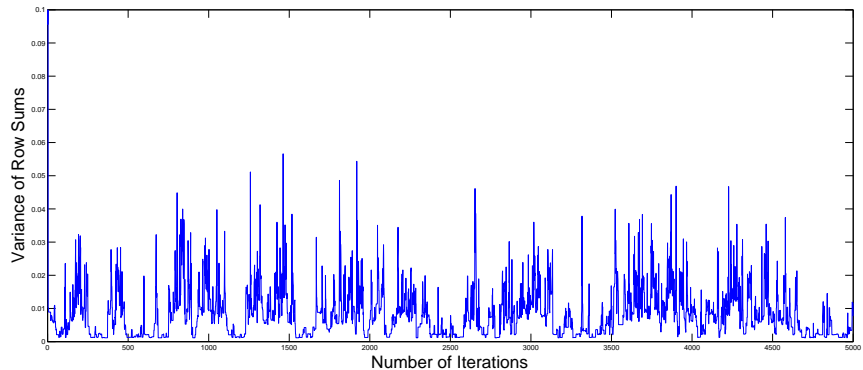}%
\end{center}
\caption{Fluctuations around a stationary distribution}
\end{figure}
%



\section{Conclusions}

This paper proposes an improved rearrangement algorithm and a stopping rule.
It can efficiently find the dependence structure that minimizes the variance
of the sum of $n$ dependent variables. It is thus able to infer the dependence
between $n-1$ variables such that the last variable is equal to the sum of the
$n-1$ first variables. As already discussed extensively in the introduction, this idea is useful in identifying the optimal
structure to achieve the Value-at-Risk bounds with a variance constraint where
the aggregate risk that maximizes and minimizes the Value-at-Risk is a
two-point distribution (\cite{Bernard-Ru-Vanduffel-2014}). This idea can also be exploited in finance to infer the joint distribution among assets for which prices of spread option or basket options are available. 

\section*{Acknowledgments}
C. Bernard thanks
the CAE research grant as well as the Alexander von Humboldt foundation and the hospitality
of the chair of Mathematical Statistics in Munich where a first draft of this paper was
completed in 2014. D. McLeish acknowledges support from NSERC. We would also like to
thank Ren\'e Carmona, Emmanuel Gobet, Giovanni Puccetti, Chris Rogers, Steven Vanduffel, Ruodu Wang and Ralf Werner for
suggestions and helpful discussions on earlier drafts of this paper.

\bibliographystyle{econometrica}
\bibliography{references}

\end{document}